%% file: root.tex
\documentclass{article}

\usepackage{arxiv}

\usepackage[pagebackref=true]{hyperref}

\usepackage{graphics} % for pdf, bitmapped graphics files
\usepackage{epsfig} % for postscript graphics files
\usepackage{mathptmx} % assumes new font selection scheme installed
\usepackage{times} % assumes new font selection scheme installed
\usepackage{amsmath} % assumes amsmath package installed
\usepackage{amssymb}  % assumes amsmath package installed
\usepackage{amsthm}
\usepackage{xypic}
\usepackage{accents}
\usepackage[ruled,lined]{algorithm2e}

\usepackage[dvipsnames]{xcolor}
\input{preamble_compatible}

\input{notation_defs}

\usepackage{stfloats}

\newcommand{\publicationdetails}{Under review}
\newcommand{\publicationversion}
{Preprint}

\providecommand{\logG}{\log_\grpG}
\providecommand{\logGv}{\log^\vee_\grpG}
\providecommand{\expG}{\exp_\grpG}
\providecommand{\pt}{\mathbf{P}}
\providecommand{\jacobi}{\mathcal{J}}

\title{\LARGE \bf
A Geometric Perspective on Fusing Gaussian Distributions on Lie Groups
}

\author{
    \href{https://orcid.org/0000-0001-7969-7039}{\includegraphics[scale=0.06]{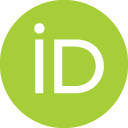}\hspace{1mm}
Yixiao Ge}
\\
    Systems Theory and Robotics Group \\
    School of Engineering \\
	Australian National University \\
    ACT, 2601, Australia \\
    \texttt{Yixiao.Ge@anu.edu.au} \\
\And    \href{https://orcid.org/0000-0003-4391-7014}{\includegraphics[scale=0.06]{orcid.png}\hspace{1mm}
Pieter van Goor}
\\
    Systems Theory and Robotics Group \\
    School of Engineering \\
	Australian National University \\
    ACT, 2601, Australia \\
    \texttt{Pieter.vanGoor@anu.edu.au} \\
	\And	\href{https://orcid.org/0000-0002-7803-2868}{\includegraphics[scale=0.06]{orcid.png}\hspace{1mm}
    Robert Mahony}
\\
    Systems Theory and Robotics Group \\
    School of Engineering \\
	Australian National University \\
    ACT, 2601, Australia \\
	\texttt{Robert.Mahony@anu.edu.au} \\
}

\begin{document}

\maketitle
% \thispagestyle{empty}
% \pagestyle{empty}

%%%%%%%%%%%%%%%%%%%%%%%%%%%%%%%%%%%%%%%%
\begin{abstract}
    Stochastic inference on Lie groups plays a key role in state estimation problems such as; inertial navigation, visual inertial odometry, pose estimation in virtual reality, etc.
    A key problem is fusing independent concentrated Gaussian distributions defined at different reference points on the group.
    In this paper we approximate distributions at different points in the group in a single set of exponential coordinates and then use classical Gaussian fusion to obtain the fused \emph{posteriori} in those coordinates.
    %
    %express all concentrated Gaussian
    %
    %In this paper we tackle the problem by approximating the Gaussians in different exponential coordinates.
    %
    %In this paper we consider the problem in terms of a change of coordinates induced by the transition functions derived from exponential coordinates.
    %
    % Stochastic inference on Lie-groups plays a key role in inertial navigation, visual inertial odometry, pose estimation in virtual reality, etc.
    % A key problem is fusing two concentrated Gaussian distributions defined at different reference points on the group.
    % State-of-the-art approaches involve non-linear optimisation along with Baker-Campbell-Hausdorff approximation of the group exponential function.
    % In this paper we consider the problem in terms of a change of coordinates induced by the transition functions derived from exponential coordinates.
    % This replaces the non-linear optimisation with (computationally simple) classical Gaussian fusion but requires approximation the Gaussian distributions in new exponential coordinates.
    We consider several approximations including the exact Jacobian of the change of coordinate map, first and second order Taylor's expansions of the Jacobian, and parallel transport with and without curvature correction associated with the underlying geometry of the Lie group.
    Preliminary results on $\SO(3)$ demonstrate that a novel approximation using parallel transport with curvature correction achieves similar accuracy to the state-of-the-art optimisation based algorithms at a fraction of the computational cost.
    % We propose a simple iterative version of the proposed algorithm that improves performance to state-of-the-art at minimal additional computational cost.
    \end{abstract}

    %%%%%%%%%%%%%%%%%%%%%%%%%%%%%%%%%%%%%%%%
    \section{Introduction}
    %%%%%%%%%%%%%%%%%%%%%%%%%%%%%%%%%%%%%%%%

    The rising interest in modern robotic and avionic systems over the past 20 years has led to increased interest in state estimation for systems defined on differentiable manifolds, and particularly on Lie groups and homogeneous spaces.
    In the Euclidean space setting, the \textit{de facto} approach to sensor fusion is that used by the (extended) Kalman filter \cite{kalmanNewApproachLinear1960}: the prior probability of the state estimate and the measurement likelihood are linearised and combined through Bayesian fusion.
    This approach is not directly compatible with the nonlinear structure of a smooth manifold and has motivated a significant body of work to adapt Bayesian fusion methodologies to the manifold setting.
    The Invariant Extended Kalman Filter (IEKF) \cite{barrauInvariantExtendedKalman2017} and Equivariant Filter (EqF) \cite{mahonyObserverDesignNonlinear2022} both contain an update step, where they pose the fusion problem in one set of exponential coordinates centered at the current state estimate, and linearise the probability functions to apply Bayesian fusion.
    In a separate work \cite{wangErrorPropagationEuclidean2006}, Wang and Chirikjian present the concept of a concentrated Gaussian distribution, which can be used to model fusion of distributions on Lie groups with exponential coordinates centered at multiple different group elements.
    This formulation is used by \cite{bourmaudIntrinsicOptimizationIterated2016,barfootAssociatingUncertaintyThreeDimensional2014}, where the authors propose an optimization algorithm to fit the posterior distribution using Bayes rule on the concentrated Gaussians directly.
    In \cite{wolfeBayesianFusionLie2011}, this same fusion problem on Lie groups is solved numerically by truncating the Baker-Campbell-Hausdorff (BCH) formula with different numbers of terms.
    Recently, in \cite{geEquivariantFilterDesign2022,geNoteExtendedKalman2023}, the authors present a new scheme that models the covariance as a tensor object on the tangent space, and uses parallel transport to compensate for the intrinsic nonlinearity of the underlying manifold.

    \begin{figure}
        \includegraphics[width=1.06\linewidth]{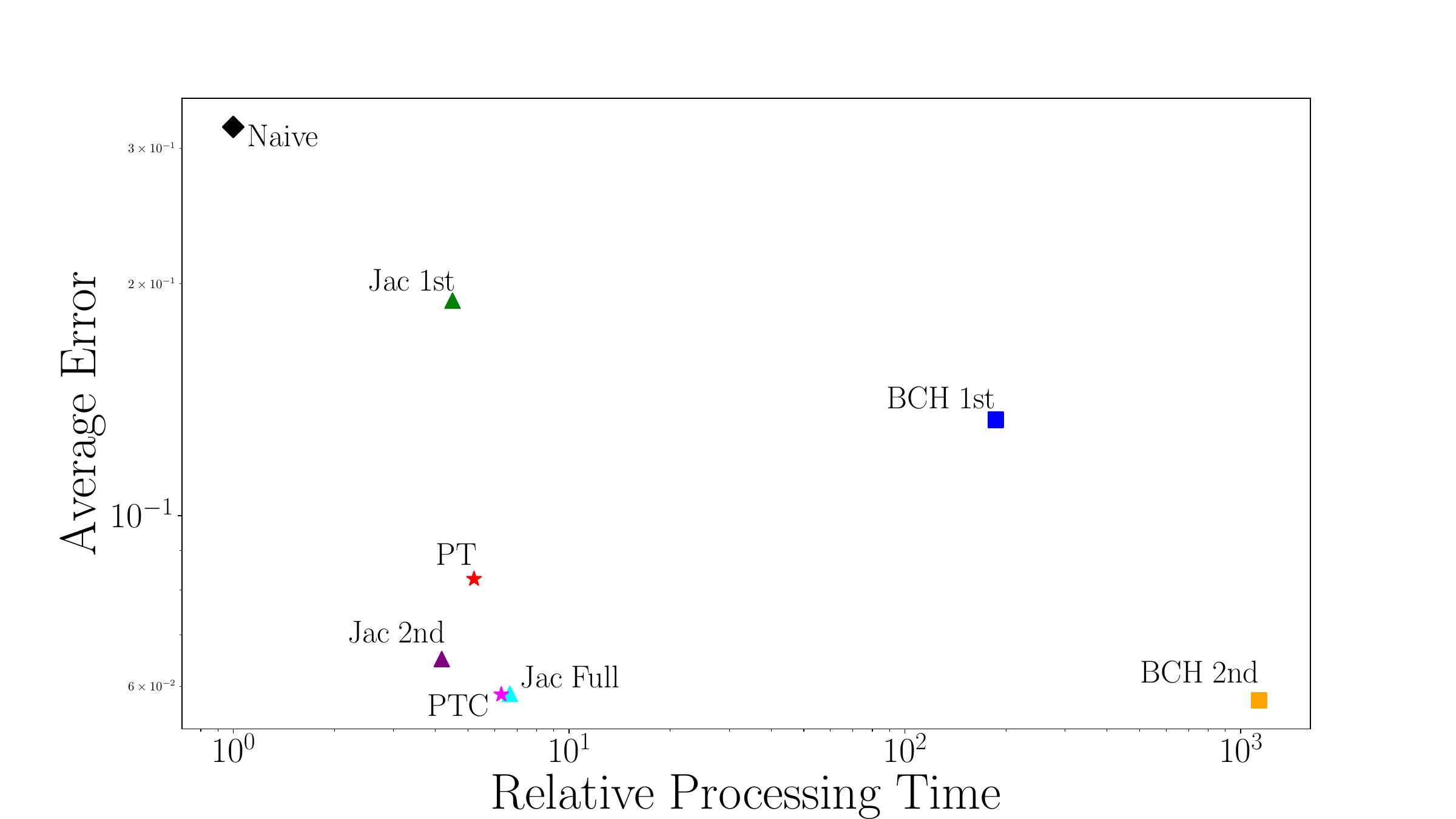}
        \centering
        % \vspace{-7mm}
        \caption{
            Comparative study of the average error~\eqref{eq:err_metric} and the relative processing time of each fusion method.
            Relative processing time is measured as a ratio of algorithm run-time with respect to run-time for the naive fusion (\textcolor{black}{$\diamond$}) algorithm.
            ``Jac Full(\textcolor{cyan}{$\triangle$})'' refers to using the analytic form of the Jacobian of the transition functions.  ``Jac 1st(\textcolor{green}{$\triangle$})'' and ``Jac 2nd(\textcolor{purple}{$\triangle$})'' refer to first and second order Taylor approximations of the Jacobian.  ``PT(\textcolor{red}{$\star$})'' refers to parallel transport and ``PTC(\textcolor{magenta}{$\star$})'' to parallel transport with curvature correction.  ``BCH 1st(\textcolor{blue}{$\square$})'' and ``BCH 2nd(\textcolor{yellow}{$\square$})'' refer to the Baker-Campbell-Hausdorff approximation methods proposed in \cite{wolfeBayesianFusionLie2011}.}
        \label{fig:main_time}
    % \vspace{-5mm}
    \end{figure}
    
    In this paper we revisit the problem of stochastic fusion of concentrated Gaussians on Lie groups.
    We extend the concept of concentrated Gaussian to allow the mean of the Gaussian to be separate from the group element at which the exponential coordinates for the distribution are centred.
    This allows us to treat Gaussian fusion on the tangent space without requiring computationally expensive optimisation procedures.
    However, it is important to compensate for the change of coordinates associated with defining a covariance of a distribution at a non-zero mean in a given set of coordinates.
    We consider several approximations including computing the exact Jacobian of the change of exponential coordinates map, first and second order Taylor's expansions of the Jacobian, and a parallel transport method.
    We also propose a novel method that does parallel transport with a curvature correction associated with the underlying geometry of the Lie group.
    This is particularly of interest since it can be implemented using only the matrix exponential function for which efficient implementations are readily available.
    We compare these five approximations with the BCH-based optimisation algorithms that are considered the state-of-art for fusion of concentrated Gaussians \cite{wolfeBayesianFusionLie2011}.
    Preliminary results on $\SO(3)$ demonstrate that using parallel transport with curvature correction achieves similar accuracy to the state-of-the-art optimisation based algorithms at a fraction of the computational cost, as shown in Fig~\ref{fig:main_time}.
    % We also propose a simple iterative update of the naive reference around which the fusion is undertaken, that achieves state-of-the-art performance at the cost of additional fusion steps.

    % \splitline
    
    % In this work, we revisit the Bayesian fusion problem on Lie groups.
    % We simplify the standard Bayesian fusion problem into a problem of fusing multiple concentrated gaussians and approximate the fused density with another concentrated Gaussian.
    % % The extended concept of concentrated Gaussian allows a better understanding of transporting
    % We show that one concentrated Gaussian may be transported to different tangent spaces via the Jacobian of the exponential mapping, which is related to the curvature of the underlying space.
    % In the situation where the Jacobian cannot be exactly derived, we propose an alternative solution by using the parallel transport and the Riemann curvature tensor as an approximation.
    % With those tools, we propose a general methodology for the fusion problem, where the zero-mean concentrated Gaussians are first transported onto the same tangent space and the distributions are then fused using the standard term-matching method on Euclidean spaces.
    % There is also an additional reset step after the fusion that transports the posterior distribution into a zero-mean concentrated Gaussian.
    % We demonstrate the performance of the proposed algorithm with a numerical simulation on the special orthogonal group $\SO(3)$.

    %%%%%%%%%%%%%%%%%%%%%%%%%%%%%%%%%%%%%%%%
    \section{Preliminaries}
    %%%%%%%%%%%%%%%%%%%%%%%%%%%%%%%%%%%%%%%%
    \subsection{General Lie groups}\label{sec:generalLie}
    Let $\grpG$ be a general Lie group with dimension $n$, associated with Lie algebra $\gothg$.
    Let $\id$ denote the identity element of $\grpG$.
    Given arbitrary $X,Y\in\grpG$, the left and right translations are denoted by $\textrm{L}_X$ and $\textrm{R}_X$, and are defined by
    \[
        \textrm{L}_X(Y):=XY, \quad \textrm{R}_X(Y):=YX.
    \]
    
    The Lie algebra $\gothg$ is isomorphic to a vector space $\R^n$ with the same dimension.
    We use the wedge $(\cdot)^\wedge:\R^n\rightarrow\gothg$ and vee $(\cdot)^\vee:\gothg\rightarrow\R^n$ operators to map between the Lie algebra and vector space.
    The Adjoint map for the group $\grpG$, $\Ad_X:{\gothg}\to{\gothg}$ is defined by
    \[
        \Ad_{X}[{{u}^{\wedge}}] = \tD \textrm{L}_{X} \circ\tD \textrm{R}_{X^{-1}}\left[{u}^{\wedge}\right] ,
    \]
    for every $X \in \grpG$ and ${{u}^{\wedge} \in \gothg}$, where $\tD \textrm{L}_{X}$, and $\tD \textrm{R}_{X}$ denote the differentials of the left, and right translation, respectively.
    Given particular wedge and vee maps, the Adjoint matrix is defined as the map ${\Ad_{X}^\vee: {\R^{n}}\to{\R^{n}}}$
    \begin{equation*}
        \Ad_{X}^\vee u = \left(\Ad_{X}{{u}^{\wedge}}\right)^{\vee} .
    \end{equation*}
    The adjoint map for the Lie algebra $\ad_{{u}^\wedge}: {\gothg}\to{\gothg}$ is given by
    \begin{equation*}
        \ad_{{u}^\wedge}{{v}^{\wedge}} = \left[{u}^{\wedge}, {v}^{\wedge}\right] ,
    \end{equation*}
    and is equivalent to the Lie bracket.
    Given particular wedge and vee maps, the adjoint matrix $\ad_u^\vee:\R^n\to\R^n$ is defined to be
    \[
        \ad_u^\vee v = \left(u^\wedge v^\wedge - v^\wedge u^\wedge\right)^\vee = \left[{u}^{\wedge}, {v}^{\wedge}\right]^\vee.
    \]
    
    Let $\expG:\gothg\to\grpG$ denote the exponential map from the Lie algebra element to the group element.
    For matrix Lie groups such as $\SO(3),\SE(3)$, this map is simply the matrix exponential.
    Let $\grpG'\subset\grpG$ be the subset of $\grpG$ where the exponential map is invertible, one can then define the logarithm map $\logG:\grpG'\to\gothg$ and $\logG^\vee:\grpG'\to\R^n$.
    For simplicity, we will suppress the wedge `$\wedge$' operator in the exponential map throughout this paper.
    
    The directional derivative of $\expG$ at a point $u\in\gothg$ in the direction of $w\in\gothg$ is given by \cite[Theorem 1.7]{helgasonDifferentialGeometryLie1979}
    \[
        \left.\ddt \expG({u+t w})\right|_{t=0}=\tD \textrm{L}_{\expG(u)}\left\{\frac{\textrm{I}-\expG({-\ad_u})}{\ad_u}(w)\right\}.
    \]
    The differential $\tD_u\expG:\gothg\to\tT_{\expG(u)}\grpG$ of $\expG$ at $u$ can be derived immediately:
    \[
        \tD_u\expG:=\tD \textrm{L}_{\expG(u)}\circ\frac{\textrm{I}-\exp({-\ad_u})}{\ad_u}.
    \]
    This map is transcendental and is computed using asymptotic expansions for a general Lie group, although for specific Lie groups such as $\SO(3)$ and $\SE(3)$, algebraic forms exist in terms of known trigonometric transcendental functions $\sin$ and $\cos$ \cite{chirikjianStochasticModelsInformation2011}.
    By identifying $\tT_{\expG(u)}\grpG$ with $\R^n$ via left trivialization and `$\vee$' operator, one can define the \emph{left trivialised} Jacobian map $J_u:\R^n\to\R^n$ as
    \begin{align}\label{eq:jacobian}
        J_u:= \tD \textrm{L}_{\expG(u)}^{-1}\circ \tD_u\expG(u) =  \frac{\textrm{I}-\exp(-\ad_u^\vee)}{\ad_u^\vee},
    \end{align}
    for any $u\in\gothg$ in the domain of $\expG$.
    The left trivialised Jacobian is referred to as the \emph{right} Jacobian in much of the key literature \cite{chirikjianStochasticModelsInformation2011}.
    The \emph{left} Jacobian can be obtained analogously by right trivialisation in the construction \eqref{eq:jacobian}.
    %
    % The inverse Jacobian map is given by \cite{hausdorffSymbolischeExponentialformelGruppentheorie2001}:
    % \[
    %     J_u^{-1}:=\sum_{k=0}^{\infty}\frac{B_k}{k!}\ad_u^{k},
    % \]
    % with the Bernoulli numbers $B_k$.

    For two functions $A, B : \R^m \to \R^n$ then we write $A(u)  = B(u) + \mathbf{o}(|u|^k)$ if $\lim_{|u| \to 0} (A(u) - B(u))/|u|^k = 0$.
    Note that this  definition is invariant under a change of coordinates.

    \subsection{Connection, curvature and parallel transport}
    For an arbitrary manifold, an affine connection is a geometric structure that is additional to the underlying differential structure.
    On Lie groups, however, there is a \emph{canonical connection} (also known as the Cartan-Schouten (0)-connection), which is defined as the only affine connection that is left-invariant, torsion-free, and has geodesics given by the Lie exponential \cite{kobayashiFoundationsDifferentialGeometry1996a}.
    % Other connections of interest in studying Lie groups include the Cartan-Schouten (-) and (+) connections \cite{kobayashiFoundationsDifferentialGeometry1996a}, which share the same left-invariance and exponential geodesics, but have non-zero torsion that results in their parallel transport maps being given by left- and right-translation, respectively.
    In the case of compact Lie groups, such as $\SO(3)$, the canonical connection coincides with the Levi-Civita connection associated with the bi-invariant metric.
    In the rest of this paper, we will consider only the Cartan-Schouten (0)-connection as the default choice of affine connection on Lie groups.

    Let $X,Y\in\gothX(\grpG)$ be two left-invariant vector fields on $\grpG$, corresponding to $x, y\in\gothg$ respectively.
    For the (0)-connection, the covariant derivative $\nabla_XY$ is given by
    \begin{align}\label{eq:connection0}
        \nabla_XY = \frac{1}{2}\left[X,Y\right].
    \end{align}
    
    The Riemann curvature tensor $\mathcal{R}:\gothX(\grpG)\times\gothX(\grpG)\times\gothX(\grpG)\to\gothX(\grpG)$ is defined by
    \[
        \mathcal{R}(X, Y) Z=\nabla_X \nabla_Y Z-\nabla_Y \nabla_X Z-\nabla_{[X, Y]} Z.
    \]
    %\begin{lemma}\label{lemma:curvature_liegroup}
    Using the Jacobi identity of the Lie bracket, it is straightforward to verify that on a Lie group $\grpG$ with (0)-connection~\eqref{eq:connection0} the Riemann curvature tensor is given by
    \begin{align}
       \mathcal{R}(X,Y)Z = -\frac{1}{4}\left[\left[X,Y\right],Z\right],
    \label{eq:curvature_liegroup}
    \end{align}
    where $X, Y, Z \in\gothX(\grpG)$ are left-invariant vector fields.
    %\end{lemma}
    %\begin{proof}
    %    The result follows directly from the skew symmetric property and the Jacobi identity of the Lie bracket.
    %\end{proof}
    
    % \subsection{Parallel transport}
    %Choose an arbitrary $u\in\gothg$, define the curve $\gamma_u(t)$ as
    %\[
    %    \gamma_u(t):=\expG(tu),
    %\]
    %which is the unique geodesic on $\grpG$ emanating from the identity to $\expG(tu)$.
    
    For the (0)-connection, the curve $\expG(tu)$ for $u \in \gothg$ arbitrary, is a geodesic.
    Given $w\in\gothg$, the parallel translation of $w$ along $\expG(tu)$ is given by
    \[
        \pt_{tu}[w] = \tD \mathrm{L}_{\expG(tu)}\circ \Ad_{\expG(-\frac{t}{2}u)}(w)\in\tT_{\expG(tu)}\grpG.
    \]
    By identifying $\tT_{\expG(tu)}\grpG$ with $\R^n$ via the left trivialisation and the `$\vee$' operator, one can define the following \emph{left trivialised parallel transport} map $\pt^\vee_u:\R^n\to\R^n$
    \[
        \pt^\vee_u:= \Ad^\vee_{\expG(-\frac{1}{2}u)}
        = \exp(-\frac{1}{2} \ad^\vee_{u})
        \in\R^{n\times n}.
    \]

    %%%%%%%%%%%%%%%%%%%%%%%%%%%%%%%%%%%%%%%%
    \section{Problem Formulation}
    %%%%%%%%%%%%%%%%%%%%%%%%%%%%%%%%%%%%%%%%
    
        In this section, we recall the concept of a concentrated Gaussian distribution on Lie groups \cite{wangErrorPropagationEuclidean2006} and its extended version \cite{geEquivariantFilterDesign2022,geNoteExtendedKalman2023}.
        We go on to formulate the fusion problem on Lie groups considered in this paper.

    \subsection{Extended Concentrated Gaussian}
    For a random variable $g\in\grpG$, the classical construction of concentrated Gaussian distribution on Lie group \cite{wangErrorPropagationEuclidean2006} is given by
    \begin{align*}
        p(g ; x, \Sigma) = \alpha \exp \left\{-\frac{1}{2}\left[\logGv\left(x^{-1}g\right)\right]^\top \Sigma^{-1} \logGv\left(x^{-1}g\right)\right\},
    \end{align*}
    where $\alpha$ is the normalising factor.
    The parameters $x\in\grpG$, $\Sigma\in\mathbb{S}_+(n)$ are the \emph{group} mean and covariance respectively, which are defined such that \cite{chirikjianStochasticModelsInformation2011}
    \begin{align*}
        \int_G \logGv\left(x^{-1} g\right) p(g) d g=\mathbf{0},
    \end{align*}
    and
    \begin{align*}
        \Sigma=\int_G\left[\logGv\left(x^{-1} g\right)\right]\left[\logGv\left(x^{-1} g\right)\right]^\top p(g) d g.
    \end{align*}
    This construction is equivalent to defining a random variable on $\grpG$ by
    \[
        g = x\expG(\epsilon), \qquad \epsilon\sim \GP(0, \Sigma)
    \]
    where $\epsilon$ is a random variable associated with a normal distribution on $\R^n$.
    
    In more recent works \cite{geEquivariantFilterDesign2022}\cite{geNoteExtendedKalman2023}, this concept was extended to allow offset mean in the Lie algebra
    \begin{align}\label{eq:ECGD}
        &p(g ; x, \mu, \Sigma) = \\
        &\alpha \exp \left\{-\frac{1}{2}\left[\logGv\left(x^{-1}g\right)-\mu^\vee\right]^\top \Sigma^{-1} \left[\logGv\left(x^{-1}g\right)-\mu^\vee\right]\right\},
    \end{align}
    where $x\in\grpG$ is termed the \emph{reference point}, $\mu\in\gothg$ is termed the \emph{mean}.
    The extended concentrated Gaussian distribution is equivalent to defining a random variable
    \[
    g = x\expG(\epsilon), \qquad \epsilon\sim \GP(\mu^\vee, \Sigma).
    \]
    The extended concentrated Gaussian makes the role of the reference point as the origin of local coordinates on the group, as separate from the mean of the underlying distribution, clear.
    Both the classical and extended concentrated Gaussian distributions are approximations of the true distributions of a random variable on a Lie group after fusion.
    The key question is not whether they are the correct model, but rather how accurately they can represent real distributions.
    
    %any random variable on a Lie-group
    %
    %Note that it is rare to have a random variable that naturally fits with a concentrated Gaussian distribution with a non-zero mean, however, it is quite common to have such distribution as the byproduct of filtering process.
    %Allowing a non-zero mean in the construction is important to understand the fusion process.

    \subsection{Fusion problem}
    
    % In this section, we propose a methodology to fuse multiple concentrated Gaussians on Lie groups.
    Consider the scenario where one has $n$ independent unbiased estimates $g \sim \GP_{x_i}(0,\Sigma_i) = p_i(g|x_i, \Sigma_i)$ for a random variable $g \in \grpG$.
    Each estimate is derived from independent data captured in the parameters  $x_i \in \grpG$ and $\Sigma_i \in \mathbb{S}_+(n)$.
    We want to derive a fused estimate $g \sim \GP_{x^+}(0,\Sigma^+)
    \approx p(g|x_1, \ldots, x_n, \Sigma_1, \ldots, \Sigma_n)$ based on all the available data.
    In classical Gaussian fusion the solution is the product of the Gaussians and can be written as a Gaussian.
    However, the product of concentrated Gaussians is not a concentrated Gaussian and the goal of this work is to find the parameters $x^+ \in \grpG$ and $\Sigma^+ \in S_+(n)$ that best approximate the fused density.
    % For simplicity, we only describe the case with two Gaussians, however, the extension to multiple Gaussians is direct.

    %%%%%%%%%%%%%%%%%%%%%%%%%%%%%%%%%%%%%%%%
    \section{Coordinate Representation of Concentrated Gaussian Distribution}
    %%%%%%%%%%%%%%%%%%%%%%%%%%%%%%%%%%%%%%%%
    
        In this section, we present the technical results including expressing the concentrated Gaussian in different reference coordinates using Jacobian mapping, and a novel approach to approximate the Jacobian with the curvature structure of the exponential map.

    \subsection{Changing reference}
    The extended concentrated Gaussian \eqref{eq:ECGD} is introduced in order to provide a model to express Gaussians around a reference that does not coincide with their mean.
    A detailed formulation is provided in the following lemma.
    
    \begin{lemma}\label{lemma:change_reference}
    Given an extended concentrated Gaussian distribution $p(g) = \GP_{x_1}(\mu_1,\Sigma_1)$ on $\grpG$ and a point $x_2 \in \grpG$ then the concentrated Gaussian $q(g) = \GP_{x_2}(\mu_2,\Sigma_2)$ with parameters
    \begin{align}
            &\mu_2 = \logG(x_2^{-1}x_1\expG(\mu_1))\label{eq:change_reference_mu}\\
            &\Sigma_2 = J_{\mu_2}^{-1}J_{\mu_1}\Sigma_1 J_{\mu_1}^\top J_{\mu_2}^{-\top}\label{eq:change_reference_cov}
    \end{align}
    minimises the Kullback-Leibler divergence of $q(g)$ with respect to $p(g)$ up to second-order linearisation error.
    \end{lemma}
    \begin{proof}
        The Kullback-Leibler divergence between $p(g)$ and $q(g)$ is given by
        \begin{align*}
            {\text{KL}}(p||q)&= \E_p[\log(p)-\log(q)]\\
                & = C_p +\frac{n}{2}\log(2\pi)+ \frac{1}{2}\log\det(\Sigma_2)\\
                &+\frac{1}{2}\E_p[\left(\logG^\vee(x_2^{-1}g)-\mu_2^\vee\right)^\top \Sigma_2^{-1}\left(\logG^\vee(x_2^{-1}g)-\mu_2^\vee\right)],
        \end{align*}
        where $C_p$ is the negative entropy of $p(g)$.
        The mean $\mu_2$ in \eqref{eq:change_reference_mu} can be derived by assuming that
        \[
            x_1\expG(\mu_1)=x_2\expG(\mu_2).
        \]
        Take the derivative of $ {\text{KL}}(p||q)$ with respect to $\Sigma_2$ and the critical point is given by
         \begin{align}\label{eq:R_prime_expect}
            \Sigma_2 = \E_p[\left(\logG^\vee(x_2^{-1}g)-\mu_2^\vee\right)\left(\logG^\vee(x_2^{-1}g)-\mu_2^\vee\right)^\top].
        \end{align}
    
        Define $\phi_1:\grpG\to\R^n$ and $\phi_2:\grpG\to\R^n$ as
        \begin{align*}
            &\phi_1(g): = \logG^\vee(x_1^{-1}g)-\mu_1^\vee,\\
            &\phi_2(g): = \logG^\vee(x_2^{-1}g)-\mu_2^\vee.
        \end{align*}
        For an arbitrary $g\in\grpG$ one has
        \begin{align}\label{eq:transfer_function}
            \phi_1(g) = \logGv(x_1^{-1}x_2\expG(\phi_2(g)+\mu_2^\vee))-\mu_1^\vee
        \end{align}
        Taking the Taylor series of \eqref{eq:transfer_function} at $\phi_2(g) = 0_{n\times 1}$ up to first order yields:
        \begin{align*}
            \phi_1(g) \approx& \tD\logG^\vee(\expG(\mu_1))\cdot\tD \textrm{L}_{x_1^{-1}x_2}\cdot\tD \expG(\mu_1) [\phi_1(g)]\\
            =&\tD\logG^\vee(\expG(\mu_1))\cdot\tD \textrm{L}_{\expG(\mu_1)}\cdot \tD\textrm{L}_{(x_1\expG(\mu_1))^{-1}}\cdot\\
            & \quad\tD\textrm{L}_{x_2\expG(\mu_2)}\cdot\tD\textrm{L}_{\expG(-\mu_2)}\cdot\tD \expG(\mu_2)\; [\phi_2(g)]\\
            =&\left(\tD\textrm{L}_{\expG(-\mu_1)}\cdot\tD\expG(\mu_1)\right)^{-1}\cdot\\
            &\quad\tD\textrm{L}_{\expG(-\mu_2)}\cdot\tD \expG(\mu_2)\; [\phi_2(g)]\\
            =& J_{\mu_1}^{-1}J_{\mu_2}\; [\phi_2(g)].
        \end{align*}
        Note that due to the choice of $\mu_2$ in \eqref{eq:change_reference_mu}, $\tD\logGv$ in the Taylor series is evaluated at $\expG(\mu_1)$.
        Substitute the result into \eqref{eq:R_prime_expect}:
        \begin{align*}
            \Sigma_2 &= \E_p[\phi_2(g)\phi_2(g)^\top]
            % &\approx\E_p[J_{\mu_2}^{-1}J_{\mu_1}\phi_1(g)(J_{\mu_2}^{-1}J_{\mu_1}\phi_1(g))^\top]\\
            % &=J_{\mu_2}^{-1}J_{\mu_1}\E_p[\phi_1(X)\phi_1(X)^\top]J_{\mu_1}^\top J_{\mu_2}^{-\top}\\
            \approx J_{\mu_2}^{-1}J_{\mu_1}\Sigma_1J_{\mu_1}^\top J_{\mu_2}^{-\top},
        \end{align*}
        which follows from the definition of $\Sigma_1 =\E_p[\phi_1(g)\phi_1(g)^\top]$.
    \end{proof}
    
    Note that \eqref{eq:change_reference_mu} is the exact coordinates of the mean in the new coordinates, and that \eqref{eq:change_reference_cov} is the covariance conjugated by the Jacobian of the change-of-coordinates maps.
    That is, Lemma \ref{lemma:change_reference} can be thought of as transforming a Gaussian distribution under a non-linear change of coordinates.
    
    \begin{remark}
        In the special case when $\mu_1=0$ or $\mu_2=0$, the covariance $\Sigma_2$ is given by
        \begin{align*}
            \Sigma_2 = J_{\mu_2}^{-1}\Sigma_1 J_{\mu_2}^{-\top}\quad \text{or}\quad \Sigma_2 = J_{\mu_1}\Sigma_1 J_{\mu_1}^\top,
        \end{align*}
        % or
        % \begin{align*}
        %     \Sigma_2 = J_{\mu_1}\Sigma_1 J_{\mu_1}^\top,
        % \end{align*}
        respectively.
        Both cases can happen in the fusion problem, as discussed in Sec~\ref{sec:method}.
    \end{remark}
    
    \subsection{Approximation with Curvature}\label{sec:PCT_approx}
    The result in Lemma~\ref{lemma:change_reference} relies on computing the linear map $J_\mu:\R^n\to\R^n$.
    However, as presented in Sec~\ref{sec:generalLie}, this map is transcendental and except in special cases, must be computed using approximations of infinite power series.
    Analytic formulae in terms of classical trigonometric functions such as $\cos$ and $\sin$ exist for a limited range of Lie groups such as $\SO(3)$ and $\SE(3)$.
    In this section we propose a method to approximate the Jacobian using geometric structure of the Lie group.
    
    % The Jacobian is directly related to the differential of the exponential mapping on $\grpG$ which induces a Jacobi field.
    % The following theorem is an application of \cite[Theorem 3.1]{langFundamentalsDifferentialGeometry2012} on Lie groups.

    \begin{theorem}\label{theorem:curv}
        For any $u\in\gothg$, the Jacobian map $J_u:\R^n\to\R^n$ defined in \eqref{eq:jacobian} satisfies
        \begin{align*}
        J_u [w^\vee] & = \pt^\vee_u[w +  \frac{1}{6}\mathcal{R}(u,w)u]^\vee + \mathbf{o}(|u^\vee|^3)\\
        &\approx \Ad^\vee_{\expG(-\frac{1}{2} u)} \left(\textrm{I}+\frac{1}{24}{\ad^\vee_u}^2\right)[w^\vee],
            % & = \Ad^\vee_{\exp(-\frac{1}{2} u)} (w+\frac{1}{24}\ad_u^2w ) = \Ad^\vee_{\exp(-\frac{1}{2} u)} (\textrm{I}+\frac{1}{24}\ad_u^2)[w].
        \end{align*}
        for all $w \in \gothg$.
    \end{theorem}
    \begin{proof}
        Define the geodesic $\gamma(t):= \expG(tu)$ with $t\in I$, where $I\subseteq \R$ is an open interval containing 0.
        Define the Jacobi field $\jacobi_w(t)$ to be the unique solution of the Jacobi equation \cite{langFundamentalsDifferentialGeometry2012}
        \[
            \tD^2_t \jacobi_w(t) + \mathcal{R}(\jacobi_w(t), \dot{\gamma}(t))\dot{\gamma}(t)=0
        \]
        for $\jacobi_w(0) = 0$  and $\tD_t\jacobi_w(0) = w$.
        For $t\in I$, one has \cite[Theorem 3.1]{langFundamentalsDifferentialGeometry2012}
        \[
            \tD_{tu}\expG[w] = \frac{1}{t}\jacobi_w(t).
        \]
        Consider the following map: $t\to(\pt_{tu})^{-1}\left(\frac{1}{t}\jacobi_w(t)\right)$.
        Note that taking the Taylor expansion of this map around $t=0$ is equivalent to studying the Taylor expansion of $(\pt_u)^{-1}\circ\tD_u\expG[w]$ around $u=0$.
        Applying \cite[Theorem A.2.9]{waldmannGeometricWaveEquations2012a} yields
        \[
        (\pt_u)^{-1}\circ\tD_u\expG [w] = w + \frac{1}{6}\mathcal{R}(u,w)u +\mathbf{o}(|u|^3).
        \]
        where the order $\order$ approximation is understood over a basis in the Lie-algebra $\gothg$.
        Take the first order approximation,
        \begin{align*}
            J_u [w^\vee] &= \tD\textrm{L}_{\expG(-u)}\circ\tD_u\expG [w]^\vee \\
            & = \pt^\vee_u[w +  \frac{1}{6}\mathcal{R}(u,w)u]^\vee +\order(|u^\vee|^3)\\
            &\approx \pt^\vee_u[w -  \frac{1}{24}[[u,w],u]]^\vee\tag{Equation~\ref{eq:curvature_liegroup}}\\
            &= \pt^\vee_u[w + \frac{1}{24}\ad_u^2(w)]^\vee\tag{Anti-symmetry of $[\cdot,\cdot]$}\\
            &= \pt^\vee_u(\textrm{I}+\frac{1}{24}{\ad^\vee_u}^2)[w^\vee]\tag{Bilinearity of $[\cdot,\cdot]$}\\
            &= \Ad_{\expG(-\frac{1}{2}u)}^\vee(\textrm{I}+\frac{1}{24}{\ad^\vee_u}^2)[w^\vee].\tag{Definition of $\pt^\vee_u$}
        \end{align*}
        This completes the proof.
        % Hence, the Jacobian $J_u$ can be approximated with
        % \begin{align}\label{eq:ptcv}
        %     J_u \approx \pt^\vee_u\:(\textrm{I}+\frac{1}{24}{\ad^\vee_u}^2),
        % \end{align}
        % where $\pt^\vee_u$ models the change of tangent spaces via the parallel transport, and $ \frac{1}{24} \pt^\vee_u \circ {\ad^\vee_u}^2$ captures the curvature structure of the exponential map.
    \end{proof}

    We show that this approximation captures enough of the necessary information of the Jacobian of the exponential map to achieve good fusion results at a low computational cost.
    
    %%%%%%%%%%%%%%%%%%%%%%%%%%%%%%%%%%%%%%%%
    \section{Fusion on Lie groups}\label{sec:method}
    %%%%%%%%%%%%%%%%%%%%%%%%%%%%%%%%%%%%%%%%
    In this section, we propose a general methodology to fuse multiple concentrated Gaussians on Lie groups.
    The proposed methodology has 3 steps.
    The first step is to compute a reference point $\hat{x} \in \grpG$.
    In the second step, one of the approximation methods is used to express the  independent concentrated Gaussians provided as data as extended concentrated Gaussians with respect to the chosen reference point.
    In these coordinates, classical Gaussian fusion is applied.
    The last step is to rewrite the fused extended concentrated Gaussian as a concentrated Gaussian around the group element corresponding to its mean.
    % \begin{List}
    % \item[Step 1: (Reference)] Compute a reference point $\hat{x} \in \grpG$.
    % \item[Step 2: (Fusion)] Write all estimates as extended concentrated Gaussians with respect to the chosen reference point.  In these coordinates compute the classical Gaussian fusion.
    % \item[Step 3: (Reset)] Rewrite the fused extended concentrated Gaussian as a concentrated Gaussian around the group element corresponding to its mean.
    % \end{List}

    % \vspace{-1mm}
    \subsection*{Step 1: Reference}
    % \vspace{-1mm}
    
    The goal of the first step of the methodology is to choose a reference point $\hat{x} \in \grpG$ as close to the true mean of the fused distribution as possible for the least reasonable computational cost.
    This point will be used as the reference point for the approximation of the independent concentrated Gaussians.
    The closer $\hat{x}$ is to the correct group-mean, the less approximation error will be incurred before the full fusion process is undertaken.
    However, spending excessive computation at this point is wasted since the independent concentrated Gaussians are defined at different points on the Lie group anyway and, as long as $\hat{x}$ is roughly central to the data, the particular choice of reference point will make little difference to the approximation.
    
    Choosing a reference point is common to most Lie-group fusion algorithms.
    In \cite{geNoteExtendedKalman2023}, and in general for Extended Kalman Filters (EKF), the mean of the prior distribution is used as reference.
    For the fusion of Gaussians, the mean $x_1$ of the first independent distribution can be used as reference.
    In \cite{wolfeBayesianFusionLie2011}, an initial estimate $\hat{x}$ is derived by using the naive fusion method.
    Such choice can be iterated to achieve better accuracy at a higher computational cost.
    The authors in \cite{barfootAssociatingUncertaintyThreeDimensional2014} use
    an optimization process to compute $\hat{x}$.
    %Sec~\ref{sec:initialchoice_discuss} and
    % \ref{sec:initialchoice_sim}.
    
    \noindent\emph{Naive Fusion: }
    We use a simple algorithm to choose $\hat{x}$ that will also act as benchmark for the comparison study in Section~\ref{sec:simulation}.
    Consider $n$ independent estimates $\GP_{x_i}(0,\Sigma_i)$ for $i = 1,\ldots, n$ of a random variable $g \in \grpG$.
    Consider exponential coordinates on the Lie group $\grpG$ around the origin.
    Approximate
    \[
    \GP_{x_i}(0,\Sigma_i) \approx \GP_{\id}(\logG(x_i),\Sigma_i)
    \]
    by an extended concentrated Gaussian in origin coordinates without any consideration of the change of coordinates on the covariance $\Sigma_i$.
    The distributions are now Gaussian in a single set of coordinates (the Lie algebra) and classical fusion is used to estimate the mean of the distribution
    \begin{align*}
    \hat{\Sigma}  = \left( \sum_{i = 1}^n \Sigma_i^{-1} \right)^{-1},
    \quad\quad\quad
    \hat{\mu}^\vee  = \hat{\Sigma} \sum_{i = 1}^n \Sigma_i^{-1} \logGv(x_i).
    \end{align*}
    Set $\hat{x} = \exp_{\grpG}(\hat{\mu})$, the final distribution is $\GP_{\hat{x}}(0,\hat{\Sigma})$.
    
    %Note that it is possible to translate, by left multiplication, the naive fusion process to compute with respect to any arbitrary point on the Lie group.
    %The identity is only used here for simplicity and to remove ambiguity in the comparisons.
    
    %We will term the random variable $g \sim \GP_{\hat{x}} (0, \hat{\Sigma})$ as the naive fusion solution and use this as a baseline for performance analysis in Section~\ref{sec:simulation}.
    
    %%%%%%%%%%%%%%%%%%%%%%%%%%%%%%%%%%%%%%%%%%%%%%%%%%%%%%%%%%%%%%%%%%%%%%%%%%%%%%%
    % \vspace{-1mm}
    \subsection*{Step 2: Fusion}
    % \vspace{-1mm}
    
    Consider independent estimates $p(g|x_i,\Sigma_i) = \GP_{x_i}(0,\Sigma_i)$
    for $i = 1, \ldots, n$.
    We approximate each distribution by an extended concentrated Gaussian
    \[
    p(g|x_i,\Sigma_i) \approx \GP_{\hat{x}}(\mu_i, \hat{\Sigma}_i)
    \]
    where $\mu_i = \logG(\hat{x}^{-1}x_i)$ and $\hat{\Sigma}_i$ is given by the chosen approximation scheme. \\
    \emph{Full Jacobian:}
    The most direct approximation is provided by applying Lemma~\ref{lemma:change_reference}:
    \begin{align}
    \hat{\Sigma}_i = J_{\mu_i}^{-1}\Sigma_i J_{\mu_i}^{-\top}
    \label{eq:Sigma}
    \end{align}
    where an analytic expression for the inverse Jacobian $J_{\mu_i}^{-1}$ can be computed. \\
    \emph{Approximate Jacobian:}
    If no analytic version of the inverse Jacobian $J_{\mu_i}^{-1}$ is available, we can approximate this by Taylors expansions.
    We consider both first and second order approximations \cite{barfootAssociatingUncertaintyThreeDimensional2014}.\\
    \emph{Parallel Transport and Curvature:}
    As discussed in Section~\ref{sec:PCT_approx} the Jacobian can be approximated by parallel transport $J_{\mu_i}^{-1} \approx \pt_{\mu_i}^{-1}$ or recalling Theorem~\ref{theorem:curv} by parallel transport and curvature
    \[
    J_{\mu_i}^{-1}
    \approx
    \left( \mathrm{I} + \frac{1}{24} \ad^2_{\mu_i}\right)^{-1}  \pt_{\mu_i}^{-1}
    \approx
    \left( \mathrm{I} - \frac{1}{24} \ad^2_{\mu_i}\right) \pt_{\mu_i}^{-1}.
    \]

    % Fig~\ref{fig:geometrictransform} demonstrates the geometric transform process.
    
    % \begin{figure}[htb!]
    % \centering
    % \includegraphics[width=\linewidth]{image/transform.png}
    % \caption{Demonstration of the geometric transform step.
    % The distributions with solid red/blue lines are the original zero-mean concentrated Gaussians.
    % The ones in dashed lines are the distributions after geometric transform.}
    % \label{fig:geometrictransform}
    % \end{figure}
    
    Once the Gaussians are written in the same coordinates the distributions are
    fused using classical Gaussian fusion.
    %\begin{subequations}\label{eq:fusion}
    \begin{align}
    \Sigma_\diamond  = \left(\sum_{i = 1}^n \hat{\Sigma}_i^{-1} \right)^{-1}, %\label{eq:mean_fusion} \\
    \quad\quad\quad
    {\mu^+}^\vee  = \Sigma_\diamond \left(\sum_{i=1}^n \hat{\Sigma}_i^{-1} \mu_i^\vee\right).
    %\label{eq:covariance_fusion}
    \label{eq:fusion}
    \end{align}
    %\end{subequations}
    Under the assumption that the distributions are independent, the fused estimate $(\mu^+,\Sigma_\diamond)$ is optimal with respect to multiple criteria such as the weighted least squares error, minimum covariance estimation error and maximum-likelihood estimation \cite{schweppeUncertainDynamicSystems1973}.
    
    % \vspace{-1mm}
    \subsection*{Step 3: Reset}\label{sec:reset}
    % \vspace{-1mm}
    
    The outcome of the fusion step is an extended concentrated Gaussian with non-zero mean $\GP_{\hat{x}}(\mu^+, \Sigma_\diamond)$.
    If a concentrated Gaussian is required, which is the normal case for most filtering algorithms and estimators, then the extended concentrated Gaussian estimate must be transformed into a concentrated Gaussian around a new group mean.
    The problem is equivalent to finding $x_+ \in \grpG$ and $\Sigma_+ \in \mathbb{S}_+(n)$ such that
    \[
    \GP_{\hat{x}}(\mu^+, \Sigma_\diamond) \approx \GP_{\hat{x}^+}(0, \Sigma^+).
    \]
    This is also a direct application of Lemma~\ref{lemma:change_reference}.
    The new reference point $\hat{x}^+$ and the covariance $\Sigma^+$ are given by
    \begin{align*}
        &\hat{x}^+ = \hat{x}\;\expG(\mu^+),\\
        &\Sigma^+ = J_{\mu^+}\:\Sigma_\diamond \:J_{\mu^+}^\top.
    \end{align*}
    The implementation of the reset will require either a full analytic version of the Jacobian to be computed, or an approximation to be used based on one of the methods discussed in Step 2.
    The new distribution $\GP_{\hat{x}^+}(0, \Sigma^+)$ is a zero-mean concentrated Gaussian.

    \section{Simulation}\label{sec:simulation}
    %%%%%%%%%%%%%%%%%%%%%%%%%%%%%%%%%%%%%%%%
    
    To numerically evaluate the proposed methods we undertake a simulation on the special orthorgonal group $\SO(3)$.
    This group provides an intuitive example for which the formula are relatively straightforward to understand.
    Following the experimental study in \cite{wolfeBayesianFusionLie2011}
    we consider two zero-mean concentrated Gaussians on $\SO(3)$, denoted by $p_1(g)=\GP_{x_1}(0,\Sigma_1)$ and $p_2(g)=\GP_{x_2}(0,\Sigma_2)$.
    The parameters are chosen as
    % \[
    %     x_1 = \expG\left(\frac{\gamma}{\sqrt{3}} \;[1.0, 1.0, -1.0]^\top\right) x_2 = \expG\left(\frac{\gamma}{\sqrt{2}} \;[1.0, -1.0, 0.0]^\top\right)
    % \]
    \small
    \[
        x_1 = \expG\left(\frac{\gamma}{\sqrt{3}} \;\begin{bmatrix}
            1.0\\1.0\\-1.0
        \end{bmatrix}\right),\quad x_2 = \expG\left(\frac{\gamma}{\sqrt{2}} \;\begin{bmatrix}
            1.0\\ -1.0\\ 0.0
        \end{bmatrix}\right)
    \]
    \normalsize
    and
    \small
    \[
        \Sigma_1 = \xi  \begin{bmatrix}
            1.0 & 0 & 0 \\
            0 & 0.75 & 0 \\
            0 & 0 & 0.5
        \end{bmatrix}, \quad
        \Sigma_2 = \xi \begin{bmatrix}
            0.5 & 0 & 0 \\
            0 & 1.0 & 0 \\
            0 & 0 & 0.75
        \end{bmatrix}.
    \]
    \normalsize
    The scalars $\gamma$, and $\xi \in \R_+$ control the distance between means and the concentration of the covariance.
    As $\gamma$ increases and $\xi$ decreases the fusion becomes more non-linear.

        To evaluate the performance, we use the cost function proposed in \cite{wolfeBayesianFusionLie2011}.
        Define $p_{12}(g):=\frac{p_1(g) p_2(g)}{\int_G p_1(h) p_2(h) \td h}$ to be the fused probability density of $p_1(g)$ and $p_2(g)$.
        We use the $\ell_1$ metric:
        \begin{align}\label{eq:err_metric}
            C:=\left|p^+-p_{12}\right|_{\ell_1}
            =\int_G\left|p^+(g)-p_{12}(g)\right| \td g,
        \end{align}
        where $p^+(g) = \GP_{\hat{x}^+}(0,\Sigma^+)$ is the estimated distribution using different methods.
        This metric evaluates the difference between the approximated concentrated Gaussian distribution and the underlying fused density, and can be interpreted as the \emph{total variation distance} between them.
    
    % To evaluate the performance, we use the cost function proposed in \cite{wolfeBayesianFusionLie2011}:
    % \begin{align}\label{eq:err_metric}
    %     E=\int_G\left|p^+(g)-\frac{p_1(g) p_2(g)}{\int_G p_1(h) p_2(h) \td h}\right| \td g
    % \end{align}
    % where $p^+(g) = \GP_{\hat{x}^+}(0,\Sigma^+)$.
    It is implemented by uniform sampling over a bounded domain on $\so(3)$.

    We run the simulation with different combinations of $\gamma$ and $\xi$, where both parameters are varied from 0.1 to 1.8.
    Each result is averaged over a Monte Carlo simulation with 500 runs
    where the covariance matrices are rotated by random rotation matrices.
    
    We present the main results demonstrating the performance of different approximation methods.
    The Naive method is implemented directly.
    Approximate Jacobian and parallel transport methods proposed in Section~\ref{sec:method} are implemented as described with the same Jacobian approximation used in the reset.
    For comparison, we include the BCH-based methods (first and second order) proposed in \cite{wolfeBayesianFusionLie2011}.
    We use the standard minimizer in the SciPy library to implement the optimisations required by these methods.
    The BCH methods do not require reset.
    All methods except the naive one use the same naive posterior as the initial guess $\hat{x}$.
    
    \begin{figure}[htb!]
        \includegraphics[width=0.7\linewidth]{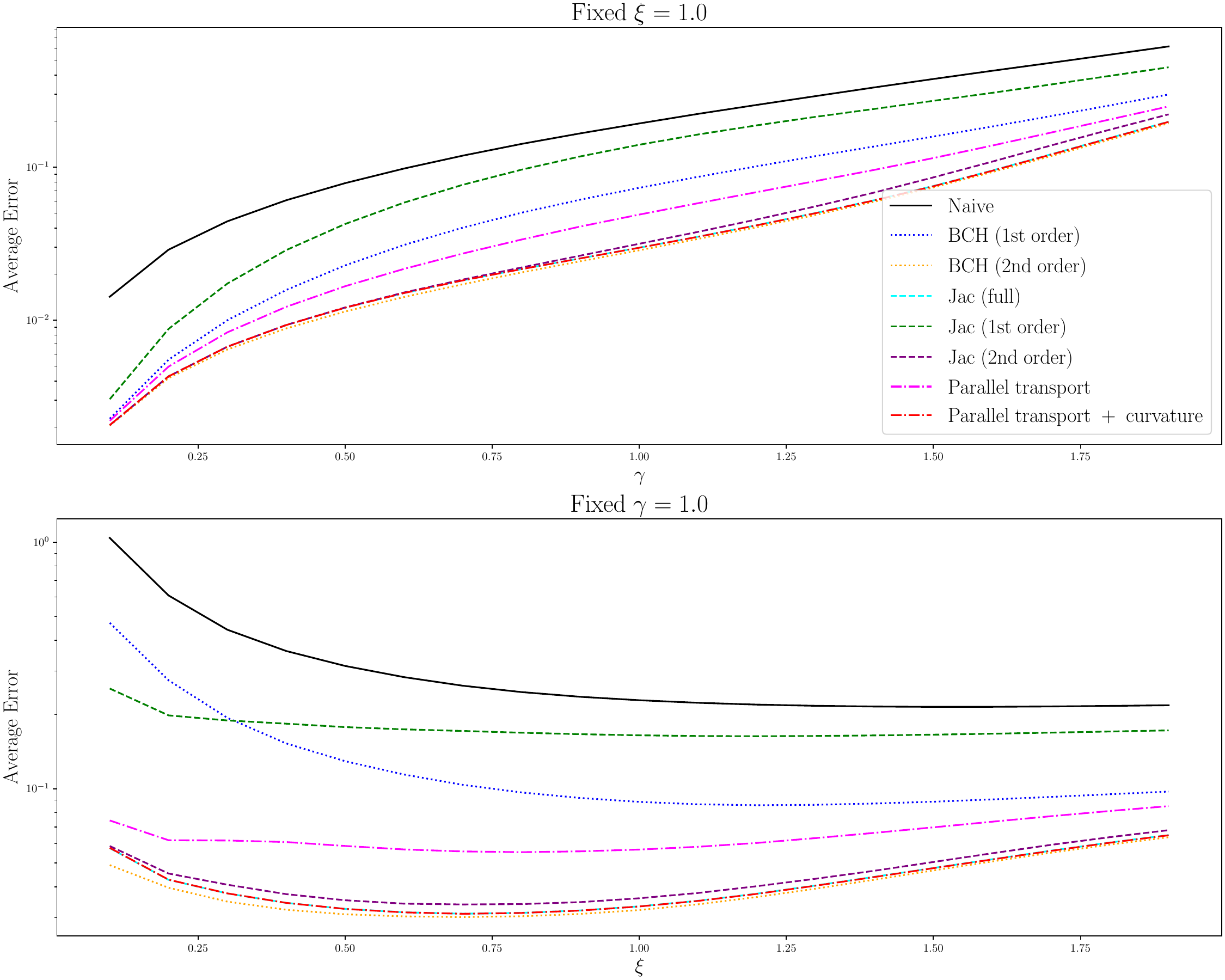}
        \centering
        % \vspace{-5.5mm}
        \caption{
            Estimation error using different approximation methods with fixed $\gamma = 1.0$ or $\xi=1.0$.
            %Note that the result for full Jacobian(\textcolor{cyan}{--}) is overlapping with the result of PTC(\textcolor{red}{--}) with difference at the scale of $10^{-3}$.
            }
        \label{fig:main_slice}
    % \vspace{-3mm}
    \end{figure}

    \begin{figure}[htb!]
        \centering
        \includegraphics[width=0.7\linewidth]{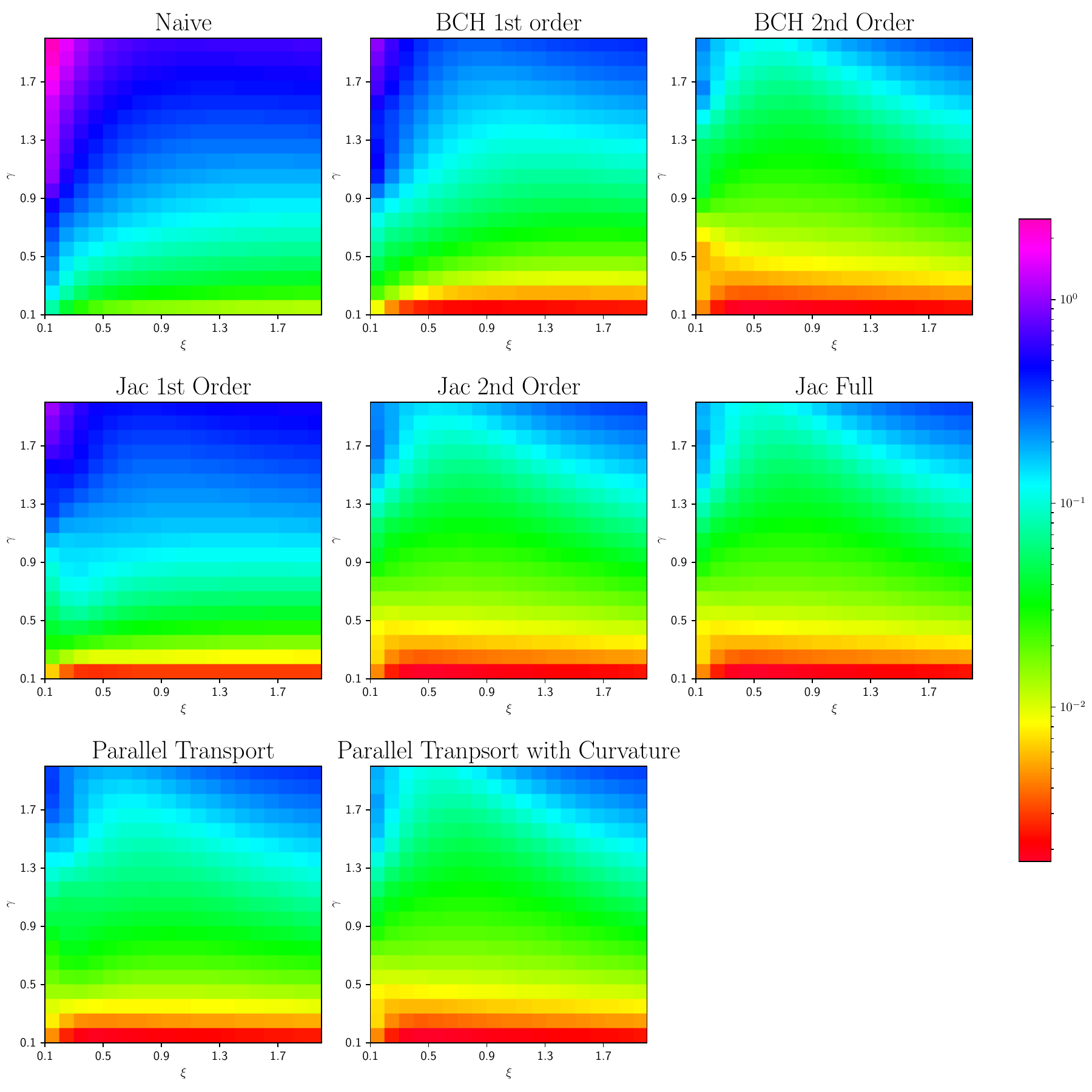}
        % \vspace{-5mm}
        \caption{
            Heatmaps showing the estimation error with different approximation methods.
            The full Jacobian and PTC perform identically, while second-order BCH method has slightly better performance for small $\gamma$.
            }
        \label{fig:main_heatmap}
    % \vspace{-5mm}
    \end{figure}
    
    Fig~\ref{fig:main_time} plots the average error against relative processing time of each method, averaged over all the parameter values of $\xi$ and $\gamma$.
    To account for the dependence on the computer hardware used, we show the ratio of processing time for each algorithm as compared to the naive fusion.
    Since all algorithms require the naive fusion as a first estimate of the reference point, this ratio is always greater than one.
    In Fig~\ref{fig:main_slice}, we present the average estimation error of different methods when either $\gamma$ or $\xi$ is fixed at 1.0 and the other parameter changes.
    Fig~\ref{fig:main_heatmap} shows how the error in each method varies with different combinations of the parameters $\gamma$ and $\xi$.
    
    In general, the naive fusion method has the worst performance among all the methods, as expected.
    The first order methods (BCH 1st order and Jac 1st order) perform marginally better but still show significant approximation error, and especially so for large values of $\gamma$.
    Interestingly, the first order BCH method is also highly sensitive to small $\xi$.
    %outperforms the first order Jacobian method in most regions as shown in Fig~\ref{fig:main_heatmap}, however, it has worse performance when both $\gamma$ and $\xi$ are small.
    The parallel transport method can also be thought of as a first order method.
    It performs better than the other first order methods but incurs more performance error than the remaining methods.
    
    The last four methods all use higher order information in the Jacobian approximation (higher order terms in the BCH expansion for the BCH method).
    Clearly, the full Jacobian outperforms the second order Taylor approximation, although the difference is less significant than the difference with the first order Jacobian.
    The full order Jacobian, parallel transport with curvature (PTC) and BCH 2nd achieve very similar average error.
    The combination of parallel transport and curvature clearly captures the major nonlinearities in the full Jacobian without requiring computation of the analytic form of the Jacobian.
    This is particularly of interest for Lie groups where a closed-form expression of the Jacobian is not available; the parallel transport only requires computing a matrix exponential, which has an efficient implementation in many linear algebra programming libraries.
    The second order BCH method achieves the lowest average error in our simulations.
    However, the BCH-based methods do not admit closed-form solutions and can only be implemented with an optimization process, resulting in much higher computational cost ($>100$ times greater than the PTC method), as shown in Fig~\ref{fig:main_time}.
    These results are based on applying a standard optimisation routine, however, even a tailored optimisation algorithm is unlikely to significantly reduce this computational discrepancy.

    \section{Conclusion}
    
    This paper presents a general design methodology for fusion of independent concentrated Gaussian distributions on Lie groups from a geometric perspective.
    It is shown that by transforming a collection of distributions into a single, unified set of coordinates, the fusion problem can be solved using the same methodology as in the Euclidean case.
    In the simulation, we present the effectiveness of the different methods proposed, and, in particular, show that the parallel transport with curvature correction can achieve good performance at a low computational cost, while using only functions that are available in most linear algebra libraries.

\bibliography{reference}
\bibliographystyle{IEEEtran}

\end{document}

%% file: preamble_compatible.tex
\usepackage{xcolor}  % Needed for colours - load before mdframed

\usepackage{amssymb}
\usepackage{amsmath}
\usepackage{mathdots}

\usepackage{mathtools}

\usepackage{tensor}

\usepackage{urwchancal}
\DeclareFontFamily{OT1}{pzc}{}
\DeclareFontShape{OT1}{pzc}{m}{it}{<-> s * [1.10] pzcmi7t}{}
\DeclareMathAlphabet{\mathpzc}{OT1}{pzc}{m}{it}

\usepackage{graphicx}
\usepackage{ifpdf}
\ifpdf
\usepackage{epstopdf}
\PrependGraphicsExtensions{.svg}
\fi
\newtheorem{theorem}{Theorem}[section]
\newtheorem{lemma}[theorem]{Lemma}

\newtheorem{remark}[theorem]{Remark}

%% file: notation_defs.tex
\providecommand{\R}{\mathbb{R}}
\providecommand{\E}{\mathbb{E}} % Euclidean space

\providecommand{\SO}{\mathbf{SO}}

\providecommand{\SE}{\mathbf{SE}}
\providecommand{\grpG}{\mathbf{G}}

\providecommand{\gothg}{\mathfrak{g}}

\providecommand{\gothX}{\mathfrak{X}} % as in X(M)

\providecommand{\so}{\mathfrak{so}}

\providecommand{\tT}{\mathrm{T}} % tangent objects eg $\tT \calM$
\providecommand{\GP}{\mathbf{N}} % Gaussian noise process.
\DeclareMathOperator{\Ad}{Ad}
\DeclareMathOperator{\ad}{ad}

\providecommand{\id}{\mathrm{id}} % identity map
\providecommand{\td}{\mathrm{d}}
\providecommand{\tD}{\mathrm{D}}

\providecommand{\ddt}{\frac{\td}{\td t}}

\usepackage{accents}
\usepackage{mathtools}
\makeatletter
\providecommand{\scirc}{%
    \hbox{\fontfamily{\rmdefault}\fontsize{0.4\dimexpr(\f@size pt)}{0}\selectfont{\raisebox{-0.52ex}[0ex][-0.52ex]{$\circ$}}}}
\makeatother
\makeatletter
\providecommand{\ucirc}{%
    \hbox{\fontfamily{\rmdefault}\fontsize{0.4\dimexpr(\f@size pt)}{0}\selectfont{\raisebox{0.0ex}[0ex][-0.52ex]{$\circ$}}}}

\makeatother

\mathchardef\mhyphen="2D
\providecommand{\order}{\mathbf{o}} % little o order

%% file: root.bbl
\begin{thebibliography}{10}
\providecommand{\url}[1]{#1}
\csname url@samestyle\endcsname
\providecommand{\newblock}{\relax}
\providecommand{\bibinfo}[2]{#2}
\providecommand{\BIBentrySTDinterwordspacing}{\spaceskip=0pt\relax}
\providecommand{\BIBentryALTinterwordstretchfactor}{4}
\providecommand{\BIBentryALTinterwordspacing}{\spaceskip=\fontdimen2\font plus
\BIBentryALTinterwordstretchfactor\fontdimen3\font minus
  \fontdimen4\font\relax}
\providecommand{\BIBforeignlanguage}[2]{{%
\expandafter\ifx\csname l@#1\endcsname\relax
\typeout{** WARNING: IEEEtran.bst: No hyphenation pattern has been}%
\typeout{** loaded for the language `#1'. Using the pattern for}%
\typeout{** the default language instead.}%
\else
\language=\csname l@#1\endcsname
\fi
#2}}
\providecommand{\BIBdecl}{\relax}
\BIBdecl

\bibitem{kalmanNewApproachLinear1960}
R.~E. Kalman, ``A {{New Approach}} to {{Linear Filtering}} and {{Prediction
  Problems}},'' \emph{Journal of Basic Engineering}, vol.~82, no.~1, pp.
  35--45, Mar. 1960.

\bibitem{barrauInvariantExtendedKalman2017}
A.~Barrau and S.~Bonnabel, ``The {{Invariant Extended Kalman Filter}} as a
  {{Stable Observer}},'' \emph{IEEE Transactions on Automatic Control},
  vol.~62, no.~4, pp. 1797--1812, Apr. 2017.

\bibitem{mahonyObserverDesignNonlinear2022}
R.~Mahony, P.~{van Goor}, and T.~Hamel, ``Observer {{Design}} for {{Nonlinear
  Systems}} with {{Equivariance}},'' \emph{Annual Review of Control, Robotics,
  and Autonomous Systems}, vol.~5, no.~1, pp. 221--252, May 2022.

\bibitem{wangErrorPropagationEuclidean2006}
Y.~Wang and G.~Chirikjian, ``Error propagation on the {{Euclidean}} group with
  applications to manipulator kinematics,'' \emph{IEEE Transactions on
  Robotics}, vol.~22, no.~4, pp. 591--602, Aug. 2006.

\bibitem{bourmaudIntrinsicOptimizationIterated2016}
G.~Bourmaud, R.~M{\'e}gret, A.~Giremus, and Y.~Berthoumieu, ``From {{Intrinsic
  Optimization}} to {{Iterated Extended Kalman Filtering}} on {{Lie Groups}},''
  \emph{Journal of Mathematical Imaging and Vision}, vol.~55, no.~3, pp.
  284--303, Jul. 2016.

\bibitem{barfootAssociatingUncertaintyThreeDimensional2014}
T.~D. Barfoot and P.~T. Furgale, ``Associating {{Uncertainty With
  Three-Dimensional Poses}} for {{Use}} in {{Estimation Problems}},''
  \emph{IEEE Transactions on Robotics}, vol.~30, no.~3, pp. 679--693, Jun.
  2014.

\bibitem{wolfeBayesianFusionLie2011}
K.~C. Wolfe and M.~Mashner, ``Bayesian {{Fusion}} on {{Lie Groups}},''
  \emph{Journal of Algebraic Statistics}, vol.~2, no.~1, Apr. 2011.

\bibitem{geEquivariantFilterDesign2022}
Y.~Ge, P.~{van Goor}, and R.~Mahony, ``Equivariant {{Filter Design}} for
  {{Discrete-time Systems}},'' in \emph{2022 {{IEEE}} 61st {{Conference}} on
  {{Decision}} and {{Control}} ({{CDC}})}, Dec. 2022, pp. 1243--1250.

\bibitem{geNoteExtendedKalman2023}
Y.~Ge, P.~Van~Goor, and R.~Mahony, ``A {{Note}} on the {{Extended Kalman
  Filter}} on a {{Manifold}},'' in \emph{2023 62nd {{IEEE Conference}} on
  {{Decision}} and {{Control}} ({{CDC}})}, Dec. 2023, pp. 7687--7694.

\bibitem{helgasonDifferentialGeometryLie1979}
S.~Helgason, \emph{Differential {{Geometry}}, {{Lie Groups}}, and {{Symmetric
  Spaces}}}.\hskip 1em plus 0.5em minus 0.4em\relax Academic Press, Feb. 1979.

\bibitem{chirikjianStochasticModelsInformation2011}
G.~S. Chirikjian, \emph{Stochastic {{Models}}, {{Information Theory}}, and
  {{Lie Groups}}, {{Volume}} 2: {{Analytic Methods}} and {{Modern
  Applications}}}.\hskip 1em plus 0.5em minus 0.4em\relax Springer Science \&
  Business Media, Nov. 2011.

\bibitem{kobayashiFoundationsDifferentialGeometry1996a}
S.~Kobayashi and K.~Nomizu, \emph{Foundations of {Differential} {Geometry},
  {Volume} 2}.\hskip 1em plus 0.5em minus 0.4em\relax John Wiley \& Sons, Feb.
  1996.

\bibitem{langFundamentalsDifferentialGeometry2012}
S.~Lang, \emph{Fundamentals of {{Differential Geometry}}}.\hskip 1em plus 0.5em
  minus 0.4em\relax Springer Science \& Business Media, Dec. 2012.

\bibitem{waldmannGeometricWaveEquations2012a}
S.~Waldmann, ``Geometric wave equations,'' \emph{arXiv preprint
  arXiv:1208.4706}, 2012.

\bibitem{schweppeUncertainDynamicSystems1973}
F.~C. Schweppe, \emph{Uncertain {{Dynamic Systems}}}.\hskip 1em plus 0.5em
  minus 0.4em\relax Prentice-Hall, 1973.

\end{thebibliography}
